\documentclass[10pt,a4paper]{article}
\usepackage[utf8]{inputenc}
\usepackage[T1]{fontenc}
\usepackage[english]{babel}
\usepackage[autostyle]{csquotes}

\usepackage{standalone}
\usepackage{amsfonts}
\usepackage{graphicx}
\usepackage{amsmath,amssymb,amsthm,geometry,verbatim,float,tikz}
\usepackage[shortlabels]{enumitem}
\usepackage{hyperref}
\usepackage{tikz}
\usepackage{caption}
\usepackage{subcaption}
\usepackage{float}
\usepackage[section]{placeins}
\usepackage{fullpage}

\renewcommand{\qedsymbol}{\rule{0.6em}{0.6em}}
\usetikzlibrary{shapes,arrows}
\usepgflibrary{shapes,arrows}
\usetikzlibrary{decorations}

\newcommand{\quotes}[1]{``#1''}

\usepgflibrary{decorations.pathreplacing}

\newtheorem{thm}{Theorem}
\newtheorem{Lemma}{Lemma}
\newtheorem{claim}{Claim}
\newtheorem{lemma}[Lemma]{Lemma}
\newtheorem{Claim}[claim]{Claim}

\newtheorem{rk}{Remark}
\newtheorem{defff}{Definition}
\theoremstyle{remark}
\newtheorem{Rk}[rk]{Remark}
\theoremstyle{definition}
\newtheorem{deff}[defff]{Definition}

\newcommand{\ZZ}{\mathbb{Z}}

\newcommand{\UU}{\mathcal{U}}
\newcommand{\A}{\mathcal{A}}
\newcommand{\BB}{\mathcal{B}}
\newcommand{\CC}{\mathcal{C}}

\newcommand{\GG}{\mathcal{G}}
\newcommand{\GGpm}{\mathcal{G^{\pm}}}
\newcommand{\GGppm}{\mathcal{G^{\pm,+-}}}

\DeclareMathOperator{\I}{I}

\usepackage{cleveref}

\usepackage{setspace}

\begin{document}

\onehalfspace

\title{Completion of the mixed unit interval graphs hierarchy}

\author{Alexandre Talon$^1$ \and Jan Kratochvíl$^2$}

\date{}

\maketitle

\begin{center}
$^1$
ENS Lyon, Lyon, France\\
\texttt{alexandre.talon@ens-lyon.org}

$^2$
Department of Applied Mathematics, Charles University, Prague, Czech Republic\\
\texttt{honza@kam.mff.cuni.cz}\\Supported by CE-ITI project GACR P202/12/G061.
\end{center}

\bigskip

\begin{abstract}
We describe the missing class of the hierarchy of mixed unit interval graphs. This class is generated by the intersection graphs of families of unit intervals that are allowed to be closed, open and left-closed-right-open. (By symmetry, considering closed, open, and right-closed-left-open unit intervals generates the same class.) We show that this class lies strictly between unit interval graphs and mixed unit interval graphs. We give a complete characterization of this new class, as well as quadratic-time algorithms that recognize graphs from this class and produce a corresponding interval representation if one exists. We also show that the algorithm from Shuchat et al. \cite{shuchat} directly extends to provide a quadratic-time algorithm to recognize the class of mixed unit interval graphs.

\bigskip

\end{abstract}
\bigskip

\section{Introduction}
A graph is an interval graph if one can associate with each of its vertices an interval of the real line such that two vertices are adjacent if and only if the corresponding intervals intersect. A well-studied subclass of the class of interval graphs is the one of proper interval graphs, generated by intersections graphs of families of unit intervals where no interval properly contains another one. This class coincides with the class of unit interval graphs~\cite{roberts}, that is when all intervals have length one.

However, in the previous descriptions no particular attention is paid to the types of intervals that are used: are they open, closed, or semi-closed?
Dourado, Le, Protti, Rautenbach and Szwarcfiter proved in \cite{dlprs} that this is of no importance as far as interval graphs are concerned. However, this is not true for unit interval graphs: deciding which types of intervals are allowed to represent the vertices of a graph is crucial. This fact was notably studied, chronogically, \cite{dlprs}, \cite{fm}, \cite{joos}, \cite{pmunit}, \cite{roberts}  and \cite{shuchat}. In these papers one can find results about the classes of graphs we can get depending on the types of unit intervals we allow for their representations. In particular, it is shown that if we require all the unit intervals used for representing a graph to be of the same type (either all closed, all open, or all semi-closed), one gets the same class of {\em unit interval graphs}. This class is a proper subclass of {\em mixed unit interval graphs} which are the graphs obtained when the intervals  are only required to be of unit length. Recently, Joos~\cite{joos} gave a characterization of mixed unit interval graphs by an infinite class of forbidden induced subgraphs, and Shuchat, Shull, Trenk and  West~\cite{shuchat} complemented it by a quadratic-time algorithm which, given any graph in this class, outputs a corresponding mixed unit interval reprensentation. In \cite{lr}, Le and Rautenbach took a different approach and studied the graphs which are representable by intervals beginning at integer positions.

The aim of this paper is to complete this hierarchy of classes of graphs representable using a specific subset of intervals of lengths one. After recalling the known results about the interval graphs hirerarchy, we show that for some $X$ the class of $\UU^X$-graphs coincide with the class of $\UU^{++}$-graphs and then show that -- with respect to this parametrization -- there exists exactly one more proper subclass of the class of mixed unit interval graphs: the class of $\UU^{\pm, +-}$-graphs. We characterize this class by an infinite list of forbidden induced subgraphs and give quadratic-time algorithms that check whether a graph belongs to this class, and in case it does, produce a corresponding appropriate interval representation.

\section{Preliminaries}
\subsection{Basic definitions and notations}
All the graphs we consider here are finite, undirected, and simple.
Let $G$ be a graph.
We denote the vertex and edge set of $G$ by $V(G)$ and $E(G)$, respectively, or by $V$ and $E$ if there are no ambiguities. We say that two vertices $u$ and $v$ are neighbors, or adjacent if $\{u,v\} \in E(G)$.

For a vertex $v\in V(G)$,
let the \emph{neighborhood} $N_G(v)$ of $v$ be the set of all vertices which are adjacent to $v$ and
let the \emph{closed neighborhood} $N_G[v]$ be defined as $N_G(v)\cup \{v\}$.
Two distinct vertices $u$ and $v$ are \emph{twins} (in $G$)
if $N_G[u]=N_G[v]$. If $G$ contains no twins, then $G$ is \emph{twin-free}.

If $C$ is a set of vertices,
then we denote by $G[C]$ the subgraph of $G$ induced by $C$.

Let $\mathcal{M}$ be a set of graphs.
We say that $G$ is $\mathcal{M}$-\emph{free}
if for every $H\in\mathcal{M}$, the graph $H$ is not an induced subgraph of $G$.

Let $\mathcal{N}$ be a family
of intervals.
We say that a graph $G$ has an $\mathcal{N}$-\emph{representation}
if there is a function $\I :V(G)\rightarrow \mathcal{N}$
such that any two distinct vertices $u$ and $v$
are adjacent if and only if $\I(u)\cap \I(v)\not=\emptyset$.
We say that $G$ is an $\mathcal{N}$-\emph{graph} if there is an $\mathcal{N}$-representation of $G$.

Let $x,y\in \mathbb{R}$. We define the \emph{closed interval} $[x,y]=\{z\in \mathbb{R}: x\leq z\leq y\}$, the \emph{open interval}
$(x,y)=\{z\in \mathbb{R}: x< z<y\}$, the \emph{closed-open interval}
$[x,y)=\{z\in \mathbb{R}: x\leq z< y\}$ and the \emph{open-closed interval}
$(x,y]=\{z\in \mathbb{R}: x< z\leq y\}$. All along this paper we draw these types of intervals as in \autoref{intervals-fig}

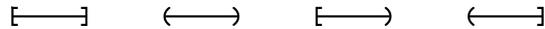
\begin{figure}[H]
\centering
\begin{tikzpicture}[yscale=0.4, thick,place2/.style={thick,
inner sep=0pt,minimum size=1mm}, place/.style={circle,draw=black,fill=black,thick,
inner sep=0pt,minimum size=1.5mm}]
\draw[{[-]}] (0,0) --(1,0);
\draw[(-)] (2,0) -- (3,0);
\draw[[-)] (4,0) -- (5,0);
\draw[{(-]}] (6,0) -- (7,0);
\draw[] (0,-1.1)--(0,-1.1);
\draw[] (0,0.3)--(0,0.3);
\end{tikzpicture}
\captionsetup{skip=0.3cm}
\caption{The closed, open, closed-open, and open-closed intervals.}
\label{intervals-fig}
\end{figure}

For an interval $A$,
let $\ell(A)=\inf(\{x\in \mathbb{R}: x\in A\})$ and $r(A)=\sup(\{x\in \mathbb{R}: x\in A\})$. We say that $A$ is a \emph{unit} interval if $r(A) = \ell(A)+1$.
If $\I$ is an interval representation of $G$ and $v\in V(G)$,
then we write $\ell(v)$ and $r(v)$ instead of $\ell(I(v))$ and $r(I(v))$
if there are no ambiguities. We refer to $\ell(A)$ and $r(A)$ as the \emph{extremities} of $A$. We say that $A$ and $B$ \emph{lie at the same position} if $\ell(A) = \ell(B)$.
We also say that a vertex $x$ has an \emph{integer interval} if $\ell(x) \in \ZZ$. We will denote such a vertex as an \emph{integer vertex}.\\
Let
$\mathcal{U}^{++}$ be the set of all closed unit intervals of the real line,
$\mathcal{U}^{--}$ be the set of all open unit intervals,
$\mathcal{U}^{-+}$ be the set of all open-closed unit intervals,
$\mathcal{U}^{+-}$ be the set of all closed-open unit intervals, and
$\mathcal{U}$ be the set of all unit intervals. We also define $\mathcal{U^{\pm}} = \mathcal{U}^{++}\cup\mathcal{U}^{--}$ and
$\displaystyle \mathcal{U}^X = \bigcup_{x \in X}{\mathcal{U}^{x}}$ for every $X \subseteq \mathcal{P}(\{++, --, -+, +-, \pm\})$. For instance, $\mathcal{U} = \mathcal{U}^{\pm, +-, -+}$.
In this terminology, $\mathcal{U}$-graphs are known as \emph{mixed unit interval graphs}. Let us also call a $\mathcal{U}^{\pm, +-}$-graph an \emph{almost-mixed unit interval graph}. We notice that by symmmetry, the class of almost-mixed unit interval graphs is also the one of $\mathcal{U}^{\pm, -+}$-graphs.

\subsection{Previous results}

First we can see that if a graph contains twins, then they can be assigned the same intervals, so in what follows we will mostly consider twin-free graphs.
We will denote by $\mathcal{G}^X$ the set of all twin-free $\mathcal{U}^X$-graphs.

We begin by recalling the known results on classifying and characterizing the unit interval graph classes.
The following two theorems characterize the most simple one.

\begin{thm}[Roberts \cite{roberts}]
A graph $G$ is a $\mathcal{U^{++}}$-graph if and only if it is a $K_{1,3}$-free interval graph.
\label{closed_unit_th}
\end{thm}

\begin{thm}[Dourado et al. \cite{dlprs}, Frankl and Maehara \cite{fm}]
The classes of $\mathcal{U}^{++}$-graphs,
$\mathcal{U}^{--}$-graphs,
$\mathcal{U}^{+-}$-graphs,
$\mathcal{U}^{-+}$-graphs, and
$\mathcal{U}^{+-, -+}$-graphs are the same.
\label{same_unit_th}
\end{thm}

The next theorem characterizes the set of twin-free graphs of the class of $\mathcal{U}^{\pm}$-graphs, that is when we allow both closed and open intervals but no others. This class is the first superclass of the class of $\mathcal{U}^{++}$-graphs.
\begin{thm}[Rautenbach and Szwarcfiter \cite{pmunit}]
\label{pm_forbidden_th}
A graph $G$ is in $\GGpm$ if and only if $G$ is a $\{K_{1,4},K_{1,4}^*,K_{2,3}^*,K_{2,4}^*\}$-free interval graph.
\end{thm}

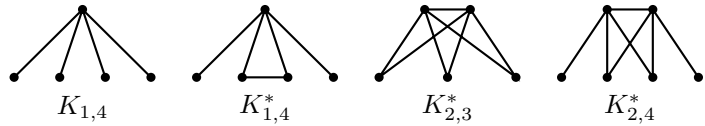
\begin{figure}[H]
\begin{center}
\begin{tikzpicture}[scale=0.6]
\def\ver{0.1} 
\def\x{1}

\def\xa{0}
\def\ya{0}

\def\xb{4}
\def\yb{0}

\def\xc{8}
\def\yc{0}

\def\xd{12}
\def\yd{0}

\path[fill] (\xa,\ya) circle (\ver);
\path[fill] (\xa+1,\ya) circle (\ver);
\path[fill] (\xa+2,\ya) circle (\ver);
\path[fill] (\xa+3,\ya) circle (\ver);
\path[fill] (\xa+1.5,\ya+1.5) circle (\ver);

\draw[thick] (\xa,\ya)--(\xa+1.5,\ya+1.5)
(\xa+1,\ya)--(\xa+1.5,\ya+1.5)
(\xa+2,\ya)--(\xa+1.5,\ya+1.5)
(\xa+3,\ya)--(\xa+1.5,\ya+1.5);

\node (1) at (\xa+1.5,\ya-0.7) {$K_{1,4}$};

\path[fill] (\xb,\yb) circle (\ver);
\path[fill] (\xb+1,\yb) circle (\ver);
\path[fill] (\xb+2,\yb) circle (\ver);
\path[fill] (\xb+3,\yb) circle (\ver);
\path[fill] (\xb+1.5,\yb+1.5) circle (\ver);

\draw[thick] (\xb,\yb)--(\xb+1.5,\yb+1.5)
(\xb+2,\yb)--(\xb+1,\yb)--(\xb+1.5,\yb+1.5)
(\xb+2,\yb)--(\xb+1.5,\yb+1.5)
(\xb+3,\yb)--(\xb+1.5,\yb+1.5);

\node (1) at (\xb+1.5,\yb-0.7) {$K_{1,4}^*$};

\path[fill] (\xc,\yc) circle (\ver);
\path[fill] (\xc+1.5,\yc) circle (\ver);
\path[fill] (\xc+3,\yc) circle (\ver);
\path[fill] (\xc+1,\yc+1.5) circle (\ver);
\path[fill] (\xc+2,\yc+1.5) circle (\ver);

\draw[thick] (\xc,\yc)--(\xc+1,\yc+1.5)--(\xc+1.5,\yc)--(\xc+2,\yc+1.5)--
(\xc+3,\yc)--(\xc+1,\yc+1.5)--(\xc+2,\yc+1.5)--(\xc,\yc);

\node (1) at (\xc+1.5,\yc-0.7) {$K_{2,3}^*$};

\path[fill] (\xd,\yd) circle (\ver);
\path[fill] (\xd+1,\yd) circle (\ver);
\path[fill] (\xd+2,\yd) circle (\ver);
\path[fill] (\xd+3,\yd) circle (\ver);
\path[fill] (\xd+1,\yd+1.5) circle (\ver);
\path[fill] (\xd+2,\yd+1.5) circle (\ver);

\draw[thick] (\xd,\yd)--(\xd+1,\yd+1.5)--
(\xd+1,\yd)--(\xd+2,\yd+1.5)--
(\xd+2,\yd)--(\xd+1,\yd+1.5)--
(\xd+2,\yd+1.5)--(\xd+3,\yd);

\node (1) at (\xd+1.5,\yd-0.7) {$K_{2,4}^*$};

\end{tikzpicture}
\end{center}
\captionsetup{skip=-0.1cm}
\caption{Forbidden induced subgraphs for twin-free $\mathcal{U}^{\pm}$-graphs.}\label{upmgraph}
\end{figure}
It is easy to see that the classes of $\mathcal{U^{\pm}}$-graphs and $\mathcal{U}^{++}$-graphs are not the same. Indeed, $K_{1,3}$ is a $\mathcal{U}^{\pm}$-graph but not a $\mathcal{U}^{++}$-graph.
A characterization of (twin-free) $\mathcal{U}$-graphs was recently given by Joos (the classes $\mathcal{R}$, $\mathcal{S}$, $\mathcal{S'}$, and $\mathcal{T}$ of forbidden induced subgraphs are depicted in Figure~\ref{graphsR}--\ref{graphsT}).

\begin{thm}[Joos \cite{joos}]
A twin-free graph $G$ is in $\GG$ if and only if $G$ is a $\{K^{*}_{2,3}\}\cup \mathcal{R} \cup \mathcal{S} \cup \mathcal{S'} \cup \mathcal{T} $-free interval graph.
\label{joos_th}
\end{thm}

\begin{figure}[H]
\centering
\begin{tikzpicture}[scale=0.6]
\def\ver{0.1} 
\def\x{1}

\def\xa{0.5}
\def\ya{0}

\def\xb{4}
\def\yb{0}

\def\xc{8}
\def\yc{0}

\def\xd{10}
\def\yd{0}

\path[fill] (\xa+0.5,\ya) circle (\ver);
\path[fill] (\xa+1,\ya+1) circle (\ver);
\path[fill] (\xa+2,\ya+1) circle (\ver);
\path[fill] (\xa+2.5,\ya) circle (\ver);
\path[fill] (\xa+1.5,\ya) circle (\ver);

\draw[thick] (\xa+0.5,\ya)--(\xa+1.5,\ya)--(\xa+1,\ya+1)
(\xa+2,\ya+1)--(\xa+1.5,\ya)--(\xa+2.5,\ya);

\node (1) at (\xa+1.5,\ya-0.7) {$R_0$};

\path[fill] (\xb,\yb) circle (\ver);
\path[fill] (\xb+1,\yb) circle (\ver);
\path[fill] (\xb+2,\yb) circle (\ver);
\path[fill] (\xb+3,\yb) circle (\ver);
\path[fill] (\xb+0.5,\yb+1) circle (\ver);
\path[fill] (\xb+1.5,\yb+1) circle (\ver);
\path[fill] (\xb+2.5,\yb+1) circle (\ver);

\draw[thick] (\xb,\yb)--(\xb+1,\yb)--(\xb+2,\yb)--(\xb+3,\yb)
(\xb+0.5,\yb+1)--(\xb+1,\yb)--(\xb+1.5,\yb+1)--(\xb+2,\yb)--(\xb+2.5,\yb+1);

\node (1) at (\xb+1.5,\yb-0.7) {$R_1$};

\path[fill] (\xd,\yd) circle (\ver);
\path[fill] (\xd+1,\yd) circle (\ver);
\path[fill] (\xd+2,\yd) circle (\ver);
\path[fill] (\xd+3,\yd) circle (\ver);
\path[fill] (\xd+4.5,\yd) circle (\ver);
\path[fill] (\xd+5.5,\yd) circle (\ver);
\path[fill] (\xd+6.5,\yd) circle (\ver);
\path[fill] (\xd+0.5,\yd+1) circle (\ver);
\path[fill] (\xd+1.5,\yd+1) circle (\ver);
\path[fill] (\xd+2.5,\yd+1) circle (\ver);
\path[fill] (\xd+5,\yd+1) circle (\ver);
\path[fill] (\xd+6,\yd+1) circle (\ver);

\fill (\xd+3.35,\yd) circle (\ver/2);
\fill (\xd+3.75,\yd) circle (\ver/2);
\fill (\xd+4.15,\yd) circle (\ver/2);

\draw[thick] (\xd,\yd)--(\xd+3,\yd)
(\xd+4.5,\yd)--(\xd+6.5,\yd)
(\xd+0.5,\yd+1)--(\xd+1,\yd)--(\xd+1.5,\yd+1)--(\xd+2,\yd)--(\xd+2.5,\yd+1)--(\xd+3,\yd)
(\xd+4.5,\yd)--(\xd+5,\yd+1)--(\xd+5.5,\yd)--(\xd+6,\yd+1);

\draw[thick,decoration={brace,mirror,raise=0.2cm},decorate] (\xd+1,\yd) -- (\xd+5.5,\yd)
node [pos=0.5,anchor=north,yshift=-0.4cm] {$i$ triangles};

\node (1) at (\xd+6,\yd-0.7) {$R_i$};





\end{tikzpicture}

\captionsetup{skip=0cm}

\caption{The class $\mathcal{R}$.}\label{graphsR}
\end{figure}
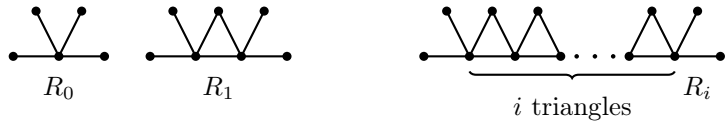

\begin{figure}[H]
 \centering

\begin{tikzpicture}[scale=0.6]
\def\ver{0.1} 
\def\x{1}

\def\xa{0}
\def\ya{0}

\def\xb{0}
\def\yb{0}

\def\xc{4}
\def\yc{0}

\def\xd{9}
\def\yd{0}

\path[fill] (\xb+1,\yb) circle (\ver);
\path[fill] (\xb+2,\yb) circle (\ver);
\path[fill] (\xb+3,\yb) circle (\ver);
\path[fill] (\xb+2,\yb+1.5) circle (\ver);
\path[fill] (\xb+1.5,\yb+1) circle (\ver);
\path[fill] (\xb+2.5,\yb+1) circle (\ver);

\draw[thick] (\xb+1,\yb)--(\xb+2,\yb)--(\xb+3,\yb)
(\xb+1,\yb)--(\xb+1.5,\yb+1)--(\xb+2,\yb)--(\xb+2.5,\yb+1)
--(\xb+2,\yb+1.5)--(\xb+2,\yb)
(\xb+1.5,\yb+1)--(\xb+2,\yb+1.5);

\node (1) at (\xb+2,\yb-0.7) {$S_1$};

\path[fill] (\xc,\yc) circle (\ver);
\path[fill] (\xc+1,\yc) circle (\ver);
\path[fill] (\xc+2,\yc) circle (\ver);
\path[fill] (\xc+3,\yc) circle (\ver);
\path[fill] (\xc+0.5,\yc+1) circle (\ver);
\path[fill] (\xc+1.5,\yc+1) circle (\ver);
\path[fill] (\xc+2.5,\yc+1) circle (\ver);
\path[fill] (\xc+1,\yc+1.5) circle (\ver);

\draw[thick] (\xc,\yc)--(\xc+1,\yc)--(\xc+2,\yc)--(\xc+3,\yc)
(\xc+0.5,\yc+1)--(\xc+1,\yc)--(\xc+1.5,\yc+1)--(\xc+2,\yc)--(\xc+2.5,\yc+1)
(\xc,\yc)--(\xc+0.5,\yc+1)--(\xc+1,\yc+1.5)--(\xc+1.5,\yc+1)
(\xc+1,\yc)--(\xc+1,\yc+1.5);

\node (1) at (\xc+1.5,\yc-0.7) {$S_2$};

\path[fill] (\xd-1,\yd) circle (\ver);
\path[fill] (\xd-0.5,\yd+1) circle (\ver);
\path[fill] (\xd,\yd) circle (\ver);
\path[fill] (\xd+1,\yd) circle (\ver);
\path[fill] (\xd+2,\yd) circle (\ver);
\path[fill] (\xd+3.5,\yd) circle (\ver);
\path[fill] (\xd+4.5,\yd) circle (\ver);
\path[fill] (\xd+5.5,\yd) circle (\ver);
\path[fill] (\xd+0.5,\yd+1) circle (\ver);
\path[fill] (\xd+1.5,\yd+1) circle (\ver);
\path[fill] (\xd+4,\yd+1) circle (\ver);
\path[fill] (\xd+5,\yd+1) circle (\ver);
\path[fill] (\xd,\yd+1.5) circle (\ver);

\fill (\xd+2.35,\yd) circle (\ver/2);
\fill (\xd+2.75,\yd) circle (\ver/2);
\fill (\xd+3.15,\yd) circle (\ver/2);

\draw[thick] (\xd-1,\yd)--(\xd+2,\yd)
(\xd+3.5,\yd)--(\xd+5.5,\yd)
(\xd-1,\yd)--(\xd-0.5,\yd+1)--(\xd,\yd)--(\xd+0.5,\yd+1)--(\xd+1,\yd)--(\xd+1.5,\yd+1)--(\xd+2,\yd)
(\xd+3.5,\yd)--(\xd+4,\yd+1)--(\xd+4.5,\yd)--(\xd+5,\yd+1)
(\xd-0.5,\yd+1)--(\xd,\yd+1.5)--(\xd+0.5,\yd+1)
(\xd,\yd)--(\xd,\yd+1.5);

\draw[thick,decoration={brace,mirror,raise=0.2cm},decorate] (\xd-1,\yd) -- (\xd+4.5,\yd)
node [pos=0.5,anchor=north,yshift=-0.4cm] {$i$ triangles};

\node (1) at (\xd+4.5,\yd-0.7) {$S_i$};

\end{tikzpicture}

\caption{The class $\mathcal{S}$.}\label{graphsS}
\end{figure}
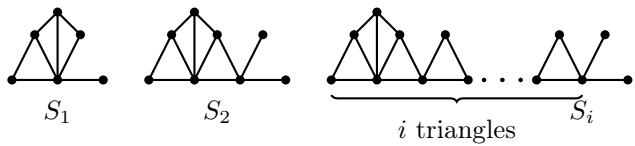

\begin{figure}[H]
\centering


\begin{tikzpicture}[scale=0.6]
\def\ver{0.1} 
\def\x{1}

\def\xa{4}
\def\ya{0}

\def\xb{0}
\def\yb{0}

\def\xc{8}
\def\yc{0}

\path[fill] (\xa,\ya) circle (\ver);
\path[fill] (\xa+1,\ya) circle (\ver);
\path[fill] (\xa+2,\ya) circle (\ver);
\path[fill] (\xa+3,\ya) circle (\ver);
\path[fill] (\xa+0.5,\ya+1) circle (\ver);
\path[fill] (\xa+1.5,\ya+1) circle (\ver);
\path[fill] (\xa+2.5,\ya+1) circle (\ver);
\path[fill] (\xa+0.5,\ya+0.5) circle (\ver);

\draw[thick] (\xa,\ya)--(\xa+1,\ya)--(\xa+2,\ya)--(\xa+3,\ya)
(\xa+0.5,\ya+1)--(\xa+1,\ya)--(\xa+1.5,\ya+1)--(\xa+2,\ya)--(\xa+2.5,\ya+1)
(\xa,\ya)--(\xa+0.5,\ya+1)--(\xa+0.5,\ya+0.5)--(\xa,\ya)
(\xa+0.5,\ya+0.5)--(\xa+1,\ya);

\node (1) at (\xa+1.5,\ya-0.7) {$S_2'$};

\path[fill] (\xb+1,\yb) circle (\ver);
\path[fill] (\xb+2,\yb) circle (\ver);
\path[fill] (\xb+3,\yb) circle (\ver);
\path[fill] (\xb+1.5,\yb+1) circle (\ver);
\path[fill] (\xb+2.5,\yb+1) circle (\ver);
\path[fill] (\xb+1.5,\yb+0.5) circle (\ver);

\draw[thick] (\xb+1,\yb)--(\xb+2,\yb)--(\xb+3,\yb)
(\xb+1,\yb)--(\xb+1.5,\yb+1)--(\xb+2,\yb)--(\xb+2.5,\yb+1)
(\xb+1.5,\yb+1)--(\xb+1.5,\yb+0.5)--(\xb+1,\yb)
(\xb+1.5,\yb+0.5)--(\xb+2,\yb);

\node (1) at (\xb+2,\yb-0.7) {$S_1'$};

\path[fill] (\xc+1,\yc) circle (\ver);
\path[fill] (\xc+2,\yc) circle (\ver);
\path[fill] (\xc+3,\yc) circle (\ver);
\path[fill] (\xc+4.5,\yc) circle (\ver);
\path[fill] (\xc+5.5,\yc) circle (\ver);
\path[fill] (\xc+6.5,\yc) circle (\ver);
\path[fill] (\xc+1.5,\yc+1) circle (\ver);
\path[fill] (\xc+2.5,\yc+1) circle (\ver);
\path[fill] (\xc+5,\yc+1) circle (\ver);
\path[fill] (\xc+6,\yc+1) circle (\ver);
\path[fill] (\xc+1.5,\yc+0.5) circle (\ver);

\fill (\xc+3.35,\yc) circle (\ver/2);
\fill (\xc+3.75,\yc) circle (\ver/2);
\fill (\xc+4.15,\yc) circle (\ver/2);

\draw[thick] (\xc+1,\yc)--(\xc+3,\yc)
(\xc+4.5,\yc)--(\xc+6.5,\yc)
(\xc+1,\yc)--(\xc+1.5,\yc+1)--(\xc+2,\yc)--(\xc+2.5,\yc+1)--(\xc+3,\yc)
(\xc+4.5,\yc)--(\xc+5,\yc+1)--(\xc+5.5,\yc)--(\xc+6,\yc+1)
(\xc+1.5,\yc+1)--(\xc+1.5,\yc+0.5)--(\xc+1,\yc)
(\xc+1.5,\yc+0.5)--(\xc+2,\yc);

\draw[thick,decoration={brace,mirror,raise=0.2cm},decorate] (\xc+1,\yc) -- (\xc+5.5,\yc)
node [pos=0.5,anchor=north,yshift=-0.4cm] {$i$ triangles};

\node (1) at (\xc+6,\yc-0.7) {$S_i'$};

\end{tikzpicture}

\caption{The class $\mathcal{S}'$.}\label{graphsS'}
\end{figure}
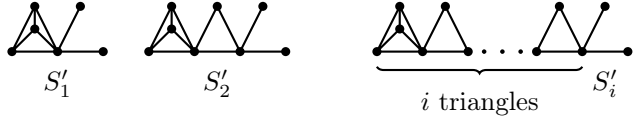

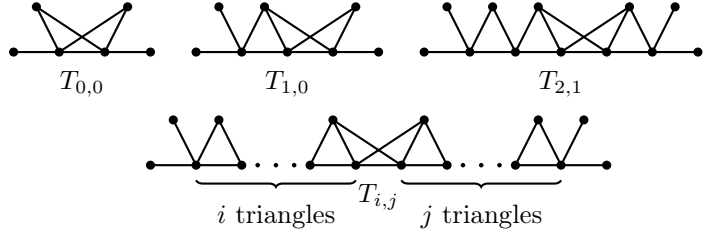
\begin{figure}[H]
\centering


\begin{tikzpicture}[scale=0.6]
\def\ver{0.1} 
\def\x{1}

\def\xa{0}
\def\ya{0}

\def\xb{4}
\def\yb{0}

\def\xc{9}
\def\yc{0}

\def\xd{3}
\def\yd{-2.5}

\path[fill] (\xa,\ya) circle (\ver);
\path[fill] (\xa+1,\ya) circle (\ver);
\path[fill] (\xa+2,\ya) circle (\ver);
\path[fill] (\xa+3,\ya) circle (\ver);
\path[fill] (\xa+0.5,\ya+1) circle (\ver);
\path[fill] (\xa+2.5,\ya+1) circle (\ver);

\draw[thick] (\xa,\ya)--(\xa+3,\ya)
(\xa+1,\ya)--(\xa+0.5,\ya+1)--(\xa+2,\ya)
(\xa+1,\ya)--(\xa+2.5,\ya+1)--(\xa+2,\ya);

\node (1) at (\xa+1.5,\ya-0.7) {$T_{0,0}$};

\path[fill] (\xb,\yb) circle (\ver);
\path[fill] (\xb+1,\yb) circle (\ver);
\path[fill] (\xb+2,\yb) circle (\ver);
\path[fill] (\xb+3,\yb) circle (\ver);
\path[fill] (\xb+4,\yb) circle (\ver);
\path[fill] (\xb+0.5,\yb+1) circle (\ver);
\path[fill] (\xb+1.5,\yb+1) circle (\ver);
\path[fill] (\xb+3.5,\yb+1) circle (\ver);

\draw[thick] (\xb,\yb)--(\xb+4,\yb)
(\xb+0.5,\yb+1)--(\xb+1,\yb)--(\xb+1.5,\yb+1)--(\xb+2,\yb)--(\xb+3.5,\yb+1)
--(\xb+3,\yb)--(\xb+1.5,\yb+1);

\node (1) at (\xb+2,\yb-0.7) {$T_{1,0}$};

\path[fill] (\xc,\yc) circle (\ver);
\path[fill] (\xc+1,\yc) circle (\ver);
\path[fill] (\xc+2,\yc) circle (\ver);
\path[fill] (\xc+3,\yc) circle (\ver);
\path[fill] (\xc+4,\yc) circle (\ver);
\path[fill] (\xc+5,\yc) circle (\ver);
\path[fill] (\xc+6,\yc) circle (\ver);
\path[fill] (\xc+0.5,\yc+1) circle (\ver);
\path[fill] (\xc+1.5,\yc+1) circle (\ver);
\path[fill] (\xc+2.5,\yc+1) circle (\ver);
\path[fill] (\xc+4.5,\yc+1) circle (\ver);
\path[fill] (\xc+5.5,\yc+1) circle (\ver);

\draw[thick] (\xc,\yc)--(\xc+6,\yc)
(\xc+0.5,\yc+1)--(\xc+1,\yc)--(\xc+1.5,\yc+1)--(\xc+2,\yc)--(\xc+2.5,\yc+1)--(\xc+3,\yc)--(\xc+4.5,\yc+1)
(\xc+5.5,\yc+1)--(\xc+5,\yc)--(\xc+4.5,\yc+1)--(\xc+4,\yc)--(\xc+2.5,\yc+1);

\node (1) at (\xc+3,\yc-0.7) {$T_{2,1}$};

\path[fill] (\xd,\yd) circle (\ver);
\path[fill] (\xd+1,\yd) circle (\ver);
\path[fill] (\xd+2,\yd) circle (\ver);
\path[fill] (\xd+3.5,\yd) circle (\ver);
\path[fill] (\xd+4.5,\yd) circle (\ver);
\path[fill] (\xd+5.5,\yd) circle (\ver);
\path[fill] (\xd+6.5,\yd) circle (\ver);
\path[fill] (\xd+8,\yd) circle (\ver);
\path[fill] (\xd+9,\yd) circle (\ver);
\path[fill] (\xd+10,\yd) circle (\ver);
\path[fill] (\xd+0.5,\yd+1) circle (\ver);
\path[fill] (\xd+1.5,\yd+1) circle (\ver);
\path[fill] (\xd+4,\yd+1) circle (\ver);
\path[fill] (\xd+6,\yd+1) circle (\ver);
\path[fill] (\xd+8.5,\yd+1) circle (\ver);
\path[fill] (\xd+9.5,\yd+1) circle (\ver);

\draw[thick] (\xd,\yd)--(\xd+2,\yd)
(\xd+3.5,\yd)--(\xd+6.5,\yd)
(\xd+8,\yd)--(\xd+10,\yd)
(\xd+0.5,\yd+1)--(\xd+1,\yd)--(\xd+1.5,\yd+1)--(\xd+2,\yd)
(\xd+8,\yd)--(\xd+8.5,\yd+1)--(\xd+9,\yd)--(\xd+9.5,\yd+1)
(\xd+3.5,\yd)--(\xd+4,\yd+1)--(\xd+4.5,\yd)
(\xd+5.5,\yd)--(\xd+6,\yd+1)--(\xd+6.5,\yd)
(\xd+4,\yd+1)--(\xd+5.5,\yd)
(\xd+6,\yd+1)--(\xd+4.5,\yd);

\fill (\xd+2.35,\yd) circle (\ver/2);
\fill (\xd+2.75,\yd) circle (\ver/2);
\fill (\xd+3.15,\yd) circle (\ver/2);

\fill (\xd+6.85,\yd) circle (\ver/2);
\fill (\xd+7.25,\yd) circle (\ver/2);
\fill (\xd+7.65,\yd) circle (\ver/2);

\draw[thick,decoration={brace,mirror,raise=0.2cm},decorate] (\xd+1,\yd) -- (\xd+4.5,\yd)
node [pos=0.5,anchor=north,yshift=-0.4cm] {$i$ triangles};

\draw[thick,decoration={brace,mirror,raise=0.2cm},decorate] (\xd+5.5,\yd) -- (\xd+9,\yd)
node [pos=0.5,anchor=north,yshift=-0.4cm] {$j$ triangles};

\node (1) at (\xd+5,\yd-0.7) {$T_{i,j}$};

\end{tikzpicture}

\caption{The class $\mathcal{T}$.}\label{graphsT}
\end{figure}

\FloatBarrier
To summarize, so far we have the following inclusions, all being proper:\\
$\{G_\varnothing\} \subsetneq \{\mathcal{U}^{++}, \mathcal{U}^{--}, \mathcal{U}^{+-}, \mathcal{U}^{-+}, \mbox{ or }
\mathcal{U}^{+-, -+}\}$-graphs $\subsetneq \mathcal{U}^{\pm}\mbox{-graphs }\subsetneq \mathcal{U}\mbox{-graphs}$, where $G_\varnothing$ is the empty graph.

However so far we have seen only 9 different sets of unit interval types, out of the 16 which exist. In the next section we will complete the picture.

\section{Our results}
In this part we take care of each of the seven missing subsets for the unit interval representations of graphs. We first consider the subsets which lead to the class of $\mathcal{U}^{++}$-graph, and then introduce the new class of almost-mixed unit interval graphs.

\subsection{Completion of the unit interval graphs hierarchy}
\begin{thm}
The classes of $\mathcal{U}^{++}$-graphs,
$\mathcal{U}^{++, +-}$-graphs,
$\mathcal{U}^{++, -+}$-graphs,
$\mathcal{U}^{--, +-}$-graphs,
$\mathcal{U}^{--, -+}$-graphs,
$\mathcal{U}^{++, +-, -+}$-graphs and
$\mathcal{U}^{--, +-, -+}$-graphs are the same.
\label{same_as_closed_unit}
\end{thm}

\begin{proof}
Firstly each of these classes contains the class of $\mathcal{U}^{++}$-graphs by \autoref{same_unit_th}.\\
Secondly, $K_{1,3}$, which is the only minimal forbidden induced subgraph for $\mathcal{U}^{++}$-graphs, is in none of these classes. Indeed, let us draw a unit interval representation of $K_{1,3}$ and show that we then need both closed and open intervals to do so. We label the vertices as in \autoref{K_1,3}. We may assume, without loss of generality, that $\ell(c) = 0$ and that $\ell(a) \leq \ell(b) \leq \ell(d)$. As all intervals have length one, their intersections enforce the following inequality: $1 = \ell(c)+1 \geq \ell(d) \geq \ell(b)+1 \geq \ell(a)+2 \geq \ell(c)+1 = 1$. This forces $\ell(a) = -1$, $\ell(b) = 0$ and $\ell(d) = 1$. It follows that $\I(c)$ must be a closed interval, the right end of $\I(a)$ must be closed and the left end of $\I(d)$ must be closed too. To meet the required intersections, $\I(b)$ must have open ends, which concludes the proof.
\end{proof}

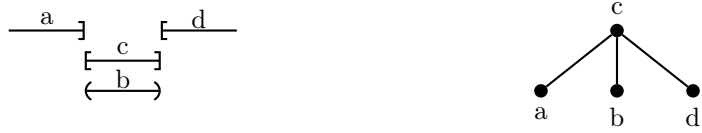
\begin{figure}[H]
\centering
\begin{tikzpicture}[yscale=0.4, thick,place2/.style={thick,
inner sep=0pt,minimum size=1mm}, place/.style={circle,draw=black,fill=black,thick,
inner sep=0pt,minimum size=1.5mm}]
\draw[{-]}] (0,2) --(1,2);
\draw[{[-]}] (1,1) -- (2,1);
\draw[(-)] (1,0) -- (2,0);
\draw[[-] (2,2) -- (3,2);
\node (a) at ( 0.5,2.4) [place2] {a};
\node (c) at ( 1.5,1.4) [place2] {c};
\node (b) at ( 1.5,0.4) [place2] {b};
\node (d) at ( 2.5,2.4) [place2] {d};

\node (cc) at ( 8,2) [place,label=c] {};
\node (aa) at ( 7,0) [place,label=below:a] {};
\node (bb) at ( 8,0) [place,label=below:b] {};
\node (dd) at ( 9,0) [place,label=below:d] {};
\draw (cc)--(bb);
\draw (aa)--(cc);
\draw (cc)--(dd);

\end{tikzpicture}

\caption{The \quotes{claw} $K_{1,3}$ and its unique $\mathcal{U}$-representations.}
\label{K_1,3}
\end{figure}

We now deal with the remaining two subsets of intervals $\mathcal{U}^{\pm, +-}$ and $\mathcal{U}^{\pm,-+}$ which lead, by symmetry, to the same class of graphs. We first show that this is a proper new class. In order to do so, we introduce a lemma about the essence of the $\mathcal{U}^{\pm, +-}$ class: the existence of an induced $K_{1,4}^*$ in every $\mathcal{U}^{\pm, +-}$-graph which is not a $\UU^{\pm}$ graph.

We call a representation {\em injective} if no two vertices are represented by the same interval. Note that every representation of a twin-free graph is injective.

\begin{lemma}
Up to symmetry, there are only two injective $\mathcal{U}$-representations of $K_{1,4}^*$, shown in \autoref{unique_K1,4*} (the leftmost interval is either open-closed or closed).
\label{unique_K1,4*_lemma}
\\
\begin{figure}[H]
\centering
\begin{tikzpicture}[yscale=0.4, thick]
\draw[{[-]}] (1,3) -- (2,3);
\draw[(-)] (1,2) --(2,2);
\draw[{-]}] (0,1) -- (1,1);
\draw[[-)] (2,1) -- (3,1);
\draw[{[-]}] (2,0) -- (3,0);
\draw[] (0,-1.1)--(0,-1.1);
\draw[] (0,3.3)--(0,3.3);
\end{tikzpicture}
\caption{The unique injective representations of $K_{1,4}^*$.}
\label{unique_K1,4*}
\end{figure}
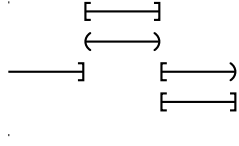
\end{lemma}

\begin{proof}
Let $\I$ be an injective $\mathcal{U}$-representation of $K_{1,4}^*$.
Let us consider one of the two $K_{1,3}$ contained in $K_{1,4}^*$. From the proof of \autoref{same_as_closed_unit}, know that it must be represented as in \autoref{K_1,3}. Now we need to add one interval for the remaining vertex. Up to symmetry, it must be at the same position as the rightmost interval of \autoref{K_1,3} and must have a closed left end. Since $\I$ is injective, we obtain the representation in \autoref{unique_K1,4*}.
\end{proof}

\begin{figure}[h]
\begin{center}
\begin{subfigure}[b]{0.3\textwidth}
\begin{tikzpicture}[scale=0.6]
\def\ver{0.1} 
\def\x{1}

\def\xa{0}
\def\ya{0}

\def\xb{4}
\def\yb{0}

\def\xc{8}
\def\yc{0}

\def\xd{12}
\def\yd{0}

\path[fill] (\xb,\yb) circle (\ver);
\path[fill] (\xb+1,\yb) circle (\ver);
\path[fill] (\xb+2,\yb) circle (\ver);
\path[fill] (\xb+3,\yb) circle (\ver);
\path[fill] (\xb+1.5,\yb+1.5) circle (\ver);
\path[fill] (\xb+1,\yb+1.5) circle (\ver);

\draw[thick] (\xb,\yb)--(\xb+1.5,\yb+1.5)
(\xb+2,\yb)--(\xb+1,\yb)--(\xb+1.5,\yb+1.5)
(\xb+2,\yb)--(\xb+1.5,\yb+1.5)
(\xb+3,\yb)--(\xb+1.5,\yb+1.5)
(\xb+1,\yb)--(\xb+1,\yb+1.5)
;

\end{tikzpicture}
\end{subfigure}
\,
\begin{subfigure}[b]{0.3\textwidth}
\begin{tikzpicture}[yscale=0.4, thick]
\draw[{[-]}] (1,3) -- (2,3);
\draw[(-)] (1,2) --(2,2);
\draw[{[-]}] (0,1) -- (1,1);
\draw[[-)] (2,1) -- (3,1);
\draw[{[-]}] (2,0) -- (3,0);
\draw[{[-]}] (3,3) -- (4,3);
\draw[] (0,-0.3)--(0,-0.3);
\draw[] (0,3.3)--(0,3.3);
\end{tikzpicture}
\end{subfigure}
\end{center}
\captionsetup{skip=0cm}
\caption{A graph which is a $\mathcal{U}^{\pm, +-}$-graph but not a $\mathcal{U}^{\pm}$-graph, and a $\mathcal{U}^{\pm, +-}$-representation of it.}
\label{not_pm}
\end{figure}
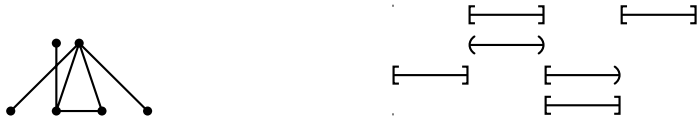

\begin{thm}
The following strict inclusions hold: $\mathcal{U}^{\pm}$-graphs $\subsetneq \mathcal{U}^{\pm, +-}$-graphs $\subsetneq \mathcal{U}$-graphs.
\end{thm}

\begin{proof}
The inclusions are immediate, we only need to show that they are strict.\\
First \autoref{not_pm} shows a graph which is a $\mathcal{U}^{\pm, +-}$-graph but not a $\mathcal{U}^{\pm}$-graph.
Now we show in \autoref{separating} a graph which is a $\mathcal{U}$-graph, but not a $\mathcal{U}^{\pm,+-}$ one.
\begin{figure}[H]
\centering
\begin{tikzpicture}[yscale=0.4, thick,place/.style={circle,draw=black,fill=black,thick,
inner sep=0pt,minimum size=1.5mm}]

\def\xa{0}
\def\ya{0}

\def\xb{0}
\def\yb{2}

\def\xc{1}
\def\yc{0}

\def\xf{1}
\def\yf{2}

\def\xd{2}
\def\yd{2}

\def\xe{2}
\def\ye{0}

\node (a) at (\xa,\ya) [place,label=below:{a}]{};
\node (b) at (\xb,\yb) [place,label={b}]{};
\node (c) at (\xc,\yc) [place,label=below:{c}]{};
\node (d) at (\xd,\yd) [place,label={d}]{};
\node (e) at (\xe,\ye) [place,label=below:{e}]{};
\node (f) at (\xf,\yf) [place,label={f}]{};
\node (h) at (\xe+1,\ye) [place,label=below:{h}]{};
\node (g) at (\xa-1,\ya) [place,label=below:{g}]{};

\draw (a)--(b)
(a)--(c)
(c)--(b)
(c)--(d)
(c)--(e)
(d)--(e)
(c)--(f)
(a)--(h)
(e)--(g);

\end{tikzpicture}
\caption{A graph separating $\mathcal{U}^{\pm, +-}$-graphs and $\mathcal{U}$-graphs.}
\label{separating}
\end{figure}
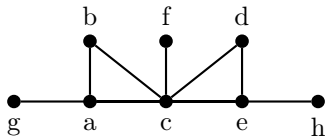
\FloatBarrier
Let us draw an injective $\mathcal{U}$-representation of this graph, and show that it is unique up to a few changes. We will see that this representation needs all four types of intervals, hence our result.\\
First we can see that it contains two induced $K_{1,4}^*$: $cfeab$ and $cfdab$. By \autoref{unique_K1,4*_lemma} and the fact that $f$ is only adjacent to $c$, $\I(c)$ and $\I(f)$ are completely determined as in \autoref{separating_repr}. Now given the neighborhoods of $a, b, d$ and $e$, and the fact that both $a$ and $e$ have one neighbor which is not adjacent to any other vertex, the intervals of $a, b, d$ and $e$ are again completely determined, up to symmetry, as in \autoref{separating_repr}. This shows that we need all four types of intervals to draw this graph.
\end{proof}

\begin{figure}[H]
\centering
\begin{tikzpicture}[yscale=0.4, thick,place2/.style={thick,
inner sep=0pt,minimum size=1mm}]
\draw[{[-]}] (1,3) -- (2,3);
\draw[(-)] (1,2) --(2,2);
\draw[{(-]}] (0,1) -- (1,1);
\draw[{[-]}] (0,0) -- (1,0);
\draw[[-)] (2,1) -- (3,1);
\draw[{[-]}] (2,0) -- (3,0);
\draw[{[-]}] (3,2) -- (4,2);
\draw[{[-]}] (-1,2) -- (0,2);
\draw [] (0,-0.3) -- (0,-0.3);

\node (a) at ( 0.5,0.4) [place2] {a};
\node (b) at ( 0.5,1.4) [place2] {b};
\node (c) at ( 1.5,3.4) [place2] {c};
\node (d) at ( 2.5,1.4) [place2] {d};
\node (e) at ( 2.5,0.4) [place2] {e};
\node (f) at ( 1.5,2.4) [place2] {f};
\node (g) at ( -0.5,2.4) [place2] {g};
\node (h) at ( 3.5,2.4) [place2] {h};

\end{tikzpicture}
\caption{A $\mathcal{U}$-representation of the graph in \autoref{separating}.}
\label{separating_repr}
\end{figure}
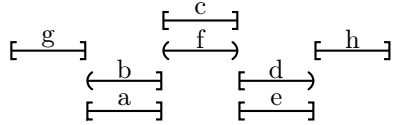

To conclude this part, we now have a complete picture of the different subclasses of the mixed unit interval class. In the schematic \autoref{classification}, $\mathcal{U}^X \subsetneq \mathcal{U}^Y$ is a shorthand notation for $\mathcal{U}^X$-graphs $\subsetneq \mathcal{U}^Y$-graphs. Sets separated by commas define the same classes of graphs.

\sloppy
\begin{figure}[H]
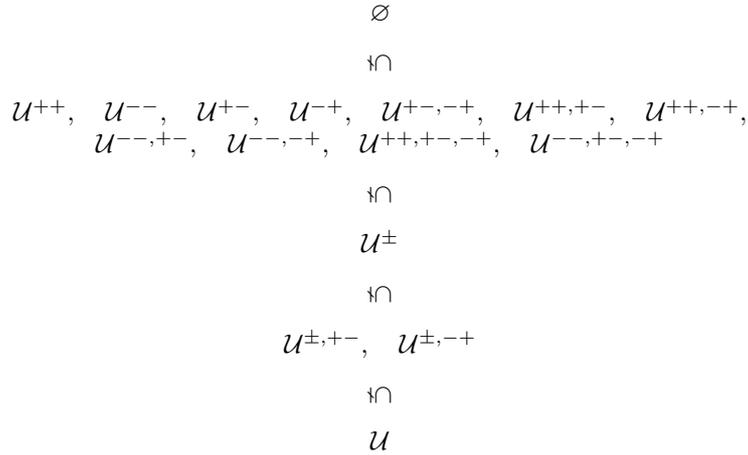

\centering $\varnothing$\\
\rotatebox{-90}{$\subsetneq$}\\
\medskip\smallskip
 $\mathcal{U}^{++} ,\quad \mathcal{U}^{--},\quad \mathcal{U}^{+-},\quad \mathcal{U}^{-+},\quad
 \mathcal{U}^{+-, -+},\quad \mathcal{U}^{++,+-},\quad \mathcal{U}^{++,-+},$\\$ \mathcal{U}^{--,+-},\quad \mathcal{U}^{--,-+},\quad \mathcal{U}^{++,+-,-+},\quad \mathcal{U}^{--,+-,-+}$\\
 \rotatebox{-90}{$\subsetneq$}\\
 \medskip\smallskip
  $\mathcal{U}^{\pm}$\\
  \rotatebox{-90}{$\subsetneq$}\\
  \medskip\smallskip
   $\mathcal{U}^{\pm,+-},\quad\mathcal{U}^{\pm,-+}$\\
   \rotatebox{-90}{$\subsetneq$}\\
   \medskip\smallskip
    $\mathcal{U}$\\
    \caption{Classification of the subclasses of the mixed unit interval graphs.}
    \label{classification}
   \end{figure}
\FloatBarrier
\subsection{Characterization of the new class: the almost-mixed unit interval graphs}

In this part, we characterize the new class of twin-free almost-mixed unit interval graphs, $\GGppm$, by a list of minimal forbidden induced subgraphs. We begin by finding which graphs may be in this list, and afterwards check that all these graphs are indeed forbidden, and minimal.
We recall that since the graphs in  $\GG$ are twin-free, any representation of such a graph is injective.

We first present a lemma which will prove to be very important in what follows. It guarantees that any graph belonging to $\GG \setminus \GG^{\pm, +-}$ has a minimal interval representation in which each semi-closed interval is \quotes{eventually} surrounded by a certain neighborhood of intervals.

\begin{deff}
Let $G \in \GG$ and $\I$ be a mixed unit interval representation of $G$.
Let $\alpha(\I)$ (resp. $\beta(\I)$)  be the number of open-closed (resp. closed-open) intervals in $\I$.
We say that $\I$ is \emph{minimal} if the couple $(\alpha(\I), \beta(\I))$ is lexicographically minimal among all other representations of $G$, that is if $\I'$ represents $G$ then either $\alpha(\I') > \alpha(\I)$ or $\alpha(\I') = \alpha(\I)$ and $\beta(\I') \geq \beta(\I)$.

\end{deff}

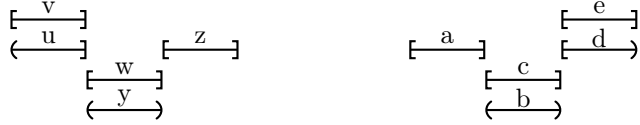
\begin{figure}[H]
\centering
\begin{subfigure}[b]{0.3\textwidth}
\begin{tikzpicture}[yscale=0.4,thick,place2/.style={thick,
inner sep=0pt,minimum size=1mm}]
\draw[{[-]}] (0,3) -- (1,3);
\draw[{(-]}] (0,2) --(1,2);
\draw[{[-]}] (1,1) -- (2,1);
\draw[(-)] (1,0) -- (2,0);
\draw[{[-]}] (2,2) -- (3,2);
\node (u) at ( 0.5,2.4) [place2] {u};
\node (v) at ( 0.5,3.4) [place2] {v};
\node (w) at ( 1.5,1.4) [place2] {w};
\node (y) at ( 1.5,0.4) [place2] {y};
\node (z) at ( 2.5,2.4) [place2] {z};
\node (yfoo) at ( 1.5,-.4) [place2] {};
\end{tikzpicture}

\caption{The neighborhood of a rightmost open-closed intervals.}
\label{neighb-u}
\end{subfigure}
\quad
\begin{subfigure}[b]{0.3\textwidth}

\begin{tikzpicture}[yscale=0.4, thick,place2/.style={thick,
inner sep=0pt,minimum size=1mm}, place/.style={circle,draw=black,fill=black,thick,
inner sep=0pt,minimum size=1.5mm}]
\draw[{[-]}] (0,2) --(1,2);
\draw[{[-]}] (1,1) -- (2,1);
\draw[(-)] (1,0) -- (2,0);
\draw[[-)] (2,2) -- (3,2);
\draw[{[-]}] (2,3) -- (3,3);
\node (a) at ( 0.5,2.4) [place2] {a};
\node (c) at ( 1.5,1.4) [place2] {c};
\node (b) at ( 1.5,0.4) [place2] {b};
\node (d) at ( 2.5,2.4) [place2] {d};
\node (d) at ( 2.5,3.4) [place2] {e};
\node (bfoo) at ( 1.5,-.4) [place2] {};
\end{tikzpicture}

\caption{The neighborhood a leftmost closed-open interval.}
\label{neighb-d}
\end{subfigure}
\caption{The neighborhood of some semi-closed intervals.}

\end{figure}

\begin{lemma}
\label{neighbours_of_u_lemma}
Let $G \in \GG$ and $\I$ be a \emph{minimal} $\mathcal{U}$-representation of it. Then
\begin{enumerate}[$(i)$]
\item if a connected component of $G$ contains a vertex whose interval is open-closed, then it contains vertices $u,v,w,y,z$ whose intervals are described in \autoref{neighb-u};
\item if a connected component of $G$ contains a vertex whose interval is closed-open, then it contains vertices $a,b,c,d,e$ whose intervals are described in \autoref{neighb-d}.
\end{enumerate}
\end{lemma}

\begin{proof}

The overall idea of the proof is the following: if, in the neighborhood of an open-closed interval, one of the mentioned intervals is missing, then we can shift some intervals and close the left end of $\I(u)$ so as to get a representation $\I'$, equivalent to $\I$, with the same number of closed-open intervals but with one less open-closed interval, hence a contradiction. It is immediate that this method also work for closed-open intervals: as we shall see, we do not create any semi-closed interval during the process but close some of then, which necessarily contradicts the minimality of I'. Therefore, we shall only prove the case of an open-closed interval, the one for closed-open interval being completely symmetrical. 
To do so, we first define
$$\varepsilon = \min(\{1\}\cup\{|x - y|: x,y \in \bigcup_{t \in V(G)}{\{\ell(t), r(t)\}} \land x \neq y\})\text{.}$$
This quantity equals the smallest \emph{non-zero} distance between any extremities of any two intervals, or 1 if the graph contains no edges. We will use  it as a security distance: it guarantees that, given an extremity of any interval, other intervals extremities can lie either at the same point or at least $\epsilon$ away from this point.\\

In this proof we say that the interval of a vertex $x$ is \emph{left-free} (resp. \emph{right-free}) if there is no other vertex $t$ such that $r(t) = \ell(x)$ (resp. $\ell(t) = r(x)$).

We begin by two useful remarks.
\begin{Rk}{Let $0 < \varepsilon' < \varepsilon$. If a vertex $x$ is such that $\I(x)$ has an open left (resp. right) end, we can either shift it by $\varepsilon'$ (resp. $-\varepsilon'$) or shift any other set of intervals by $-\varepsilon'$ (resp. $\varepsilon'$) without losing any intersection involving $\I(x)$ (but we can gain intersections).}
\label{open_end}
\end{Rk}
This comes from the definition of $\varepsilon$: since the left end of $\I(x)$ is open, any interval intersecting it at its left must do it on more than a single point, hence the intersection is of length at least $\varepsilon > \varepsilon'$.

\begin{Rk}{Let $\I(x)$ be a left-free (resp. right-free) interval. Closing its left (resp. right) end does not create any intersection.}
\label{left-free}
\end{Rk}

\begin{Claim}{If an interval $I(s)$ is semi-closed, then there exists some closed $\I(t)$ at the same position.}
\label{claim1}
\end{Claim}
\begin{proof}[Proof of \autoref{claim1}]
We first deal with the case when $\I(s)$ is open-closed, and assume for convenience, up to translating the whole interval representation, that $\ell(s) = 0$. We suppose to the contrary that there is no such $\I(t)$. We would like to close the left end of $\I(s)$. To do so, let us define $\I'$ in the following way:
\begin{itemize}
\item $\I'(x) = \I(x)-\varepsilon/2$ if $x \neq s$, $\ell(\I(x)) \in \ZZ$ and $\ell(\I(x)) \leq 0$;
\item $\I'(s)$ = $[0,1]$ (now it is closed);
\item $\I'(x) = \I(x)$, otherwise.
\end{itemize}
We show that $\I$ and $\I'$ are equivalent, that is they represent the same graph. By the definition of $\varepsilon$, we modify no intersection involving any non-integer interval. Since we do not shift the intervals beginning from $1$ on, and we shift all integer intervals $J \in \I(G)$ such that $\ell(J) \leq 0$ by the same quantity, the only intersections we can change involve $\I(s)$ or an interval at the same position as $\I(s)$. Since $\I$ is injective and there is no $[0,1]$ interval, any interval sharing the position of $\I(s)$ must have an open right end. Therefore, it had no intersection at $1$, and shifting it does not remove any intersection. The same applies for $\I(s)$: since its left end is open, it does not lose any intersection. Moreover, since we shifted all other integer intervals, we can close it without creating any new intersection.\\
This shows the equivalence between $\I$ and $\I'$, which contradicts the minimality of $\I$.\\
If $\I(s)$ is closed-open, we proceed in a symmetric way. The $\I'$ we obtain contains one less closed-open interval, which still contradicts the minimality of $\I$.
\renewcommand{\qedsymbol}{$\Box$}
\end{proof}

Now, let $C$ be the connected component we consider, which contains an open-closed interval.
We first define $u$ to be the vertex such that $\I(u)$ is the rightmost open-closed interval in the component $C$. This choice will prove to be essential when we prove the existence of $y$. Up to translating the whole interval representation, we will conveniently assume that $\ell(u) = 0$ throughout the proof.

The existence of $v$ (such that $\I(v) = [0,1]$) is directly given by \autoref{claim1}.\\

Now we deal with the existence of $w$. We again proceed by contradiction, and suppose that no such interval exists. We define $\I'$ as follows:
\begin{itemize}
\item $\I'(t) = \I(t)-\varepsilon/2$ if $t \neq u$, $\ell(\I(t)) \in \ZZ$ and $\ell(\I(t)) \leq 0$;
\item $\I'(u)$ = $[-\varepsilon/4,1-\varepsilon/4]$;
\item $\I'(t) = \I(t)$, otherwise.
\end{itemize}
Using the same arguments as in \autoref{claim1}, we conclude that the first line of the definition of $\I'$ preserves all the intersections and creates none, except possibly the ones with $[1,2]$ or $[1,2)$. However, by assumption there is no $[1,2]$ interval and then by the contrapositive of~\autoref{claim1}, there is no $[1,2)$ interval, so in $\I'$ $v$ also keeps exactly the intersections it has in $\I$. For the same reason, shifting $\I(u)$ by $-\varepsilon/4$ removes no intersections at its right. Since we shift it by less than the other intervals, it is now left-free, and so \autoref{left-free} guarantees that by closing its left end we create no intersection.\\Therefore, $\I$ and $\I'$ are equivalent, which is a contradiction to the minimality of $\I$. So we proved the existence of a vertex $w$ such that $\I(w) = [1,2]$.\\

Our next step is to prove the existence of $\I(y)$.
We suppose that there exists no such $(1,2)$ interval. The choice of $u$ we made tells us that there is no $(1,2]$ interval either. We then define $\I'$ as follows:
\begin{itemize}
\item $\I'(u)$ = $[\varepsilon/2$, $1+\varepsilon/2]$;
\item $\I'(t) = \I(t)$, otherwise.
\end{itemize}
The interval representations $\I$ and $\I'$ are equivalent: since there is no interval with an open left end at 1, shifting $\I(u)$ does not make it gain any intersection. By \autoref{open_end}, it loses none at its left. Furthermore, by definition of $\varepsilon$, $\I(u)+\varepsilon/2$ is left-free, so by \autoref{left-free} we can close it without adding any intersection.\\This contradicts the minimality of $\I$, and provied us with a vertex $y$ such that $\I(w) = (1,2)$.\\

We now show the existence of $\I(z)$.\\
We proceed again by contradiction: if there is no such $[2,3]$ interval, then we can define $\I'$ as follows:
\begin{itemize}
\item $\I'(t) = \I(t)+\varepsilon/2$ if $\ell(\I(t)) \in \ZZ$, $\ell(\I(t)) \geq 2$;
\item $\I'(y) = (1+\varepsilon/2,2+\varepsilon/2)$;
\item $\I'(u) = [\varepsilon/2, 1+\varepsilon/2]$;
\item $\I'(t) = \I(t)$, otherwise.
\end{itemize}
We show that $\I$ and $\I'$ are equivalent. Since there is no $[2,3]$ interval, by the contrapositive of \autoref{claim1} there is no $[2,3)$ interval, hence by the same arguments as in the proof of \autoref{claim1}, we lose no intersection by the first line of the definition of $\I'$. Owing to the first shift and the definition of $\varepsilon$, shifting $\I(y)$ does not create any intersection at its right. Since its left end is open, \autoref{open_end} guarantees that we lose no intersection at its left. Since $\I(u)$ has an open left end, shifting it modifies no intersection at its left. Since $G$ is twin-free, we have shifted $\I(y)$ and there is no $(1,2]$ interval, we create no intersection at its right. Besides, $\I'(u)$ is now left-free, hence we can close it. This contradicts again the minimality of $\I$.
\end{proof}

Now we look for all possible forbidden induced minimal subgraphs of any $G \in \GG \setminus \GGppm$. In what follows, we denote by the class $\mathcal{A}$ the union $\cup{A_i}$, and so on.

\begin{lemma}
Let $G \in \GG$. If $G \notin \GGppm$ then it contains an induced copy of a graph in $\A \cup \BB \cup \BB' \cup \BB'' \cup \CC \cup \CC'$ (see \autoref{class_A} to \autoref{class-C'}).
\label{notGGpm1}
\end{lemma}

\begin{proof}
Let us take such a graph $G \in \GG \setminus \GGppm$ and consider $\I$ a \emph{minimal} $\mathcal{U}$-representation of $G$, that is one with minimum number of open-closed intervals, and subject to this condition, minimum number of closed-open intervals.

First, since $G \in \GG$ and $G \notin \mathcal{G}^{\pm, +-}$, there exist one connected component containing both an  open-closed interval and a closed-open interval. Indeed, if this is not the case, by symmetrizing the intervals in all the components containing no closed-open interval, we obtain an interval representation of $G$ containing no open-closed interval, which is a contradiction.

So from now on we assume that we have vertices $u$ and $d$ in a same connected component such that, from \autoref{neighbours_of_u_lemma}, $\I(u)$ is open-closed, $\I(d)$ is closed-open, and they come with vertices $v,w,y,z,a,b,c,e$ whose intervals are the one of \autoref{neighb-u} and \autoref{neighb-d}. We also assume, up to translating the whole interval representation, that $\ell(a) = 0$. We now consider all possible values for $\ell(u)$. By doing so, we establish a list of graphs among which $G$ must have an induced copy, since it is not in $\GGppm$. 

\input{big_list_forbidden.tex}
We begin by the case when $-2 \leq \ell(u) < 3$. We first consider the subcase when $\I(u)$ is an integer vertex, which implies that some intervals are the same. Since the graph is supposed twin-free, it means that these vertices are the same. This case is covered by the class $\mathcal{C}$ (\autoref{class-C}), where $C_i$ means that $\ell(u) = i$.

Now we still consider the case when $-2 \leq \ell(u) < 3$ but we examine the other subcase, that is when $\I(u)$ is not an integer interval. In this case, since $a,b,c,d,e$ are not integer vertices, we only need to consider in which integer interval the beginning of $\I(u)$ strictly lies in. For instance, it makes no difference if $\I(u) = 0.5$ or $\I(u) = 0.6$ since the graph represented is the same. This case is covered by the class $\mathcal{C'}$ (\autoref{class-C'}). $C'_i$ represents the case when $i < \ell(u) < i+1$.

We now consider the remaining cases, that is when no interval for the vertices $a,b,c,d,e$ intersects an interval for the vertices $u,v,w,y,z$. Here we can notice that the fact that whether $\I(u)$ is integral or not has no interest. The case when $\ell(u) < 2$ corresponds to the class $\mathcal{A}$ (\autoref{class_A}), and the remaining case, when $\ell(u) \geq 3$ corresponds to the classes $\mathcal{B}$, $\mathcal{B'}$ and $\mathcal{B''}$ (\autoref{class_B}, \autoref{class_B'}, \autoref{class_B''}). One could argue that these three classes do not cover the previous cases, because the subgraph induced by the $p_i$'s (the vertices between $z$ and $a$ for class $\mathcal{A}$) could have more edges than a path. However, if that were the case then we could remove some of the $p_i$'s to get a shorter path and our new graph would be an induced subgraph of the former one.
\end{proof}

\begin{lemma}
Let $G \in \GG$. If $G$ contains an induced copy of a graph in $\A \cup \BB \cup \BB' \cup \BB'' \cup \CC \cup \CC'$ then $G \in \GG \setminus \GGppm$.
\label{notGGpm2}
\end{lemma}

\begin{proof}
First, we justify the fact that the classes $\mathcal{B}$, $\mathcal{B'}$ and $\mathcal{B''}$ are forbidden. This is because the graphs in these classes contain the pattern induced by the vertices $a,b,c,d,e,u,v,w,y,z$.
Indeed, \autoref{unique_K1,4*_lemma} specifies that the two copies of $K_{1,4}^*$ these vertices form must be represented, up to symmetry, as in \autoref{unique_K1,4*}. Since there is a path between $e$ and $v$, which is vertex-disjoint from $d$ and $u$, the two interval representations must be symmetrical, hence the need for the two types of semi-closed intervals.

For the class $\mathcal{A}$, we have again the same two copies of $K_{1,4}^*$ and their interval configurations shown in \autoref{unique_K1,4*}, but here vertices $a$ and $z$ are connected by a path which is vertex-disjoint from the two $K_{1,4}^*$, so these two occurrences must be symmetrical, hence these graphs are forbidden.

For the graphs $C'_{-2}$, $C'_{-1}$, $C'_0$, $C'_1$ and $C'_2$ the point is that we have two vertex-disjoint $K_{1,4}^*$ ($decba$ and $uvwyz$). By \autoref{unique_K1,4*_lemma} we know that they can be represented by only two sets of intervals. However if we begin to draw the intervals for $decba$, then there is only one choice for $uvwyz$, up to a small translation (that is, it is equivalent whether one interval begins at 1.4 ou 1.5 for instance).\\
For the graphs $C_{-2}$, $C_{-1}$, $C_0$, $C_1$ and $C_2$ the argument is the same, except that the two $K_{1,4}^*$ share some vertices. We first begin to draw $decba$, and then realize that the other intervals must be exactly as in the above figures.
\end{proof}

We now state our main theorem.
\begin{thm} A twin-free graph $G$ is in $\GGppm$ if and only if it is a $\mathcal{A}\cup\mathcal{B}\cup\mathcal{B'}\cup\mathcal{B''}\cup
\mathcal{C}\cup\mathcal{C'}\cup\mathcal{S}\cup\mathcal{S'}\cup\{T_{0,j}: j \geq 0\}\cup\{T_{1,1}, R_0, R_1, K_{2,3}^*\}$-free interval graph.
\label{bigth}
\end{thm}

\begin{proof}
Since $G$ is in $\GGppm$ if and only if $G \in \GG$ and $G \in \GG \setminus \GGpm$, we know by combining \autoref{joos_th}, \autoref{notGGpm1} and \autoref{notGGpm2} that $G \in \GGppm$ if and only if it is a $\{K^{*}_{2,3}\}\cup \mathcal{R} \cup \mathcal{S} \cup \mathcal{S'} \cup \mathcal{T}\mathcal{A}\cup\mathcal{B}\cup\mathcal{B'}\cup\mathcal{B''}\cup
\mathcal{C}\cup\mathcal{C'}$-free interval graph.
However, we can notice that from the class $\mathcal{R}$ we only need $R_0$ and $R_1$ since the other ones are supergraphs of graphs in $\mathcal{B}$, hence already forbidden. We need $K_{2,3}^*$ and all the graphs in $\mathcal{S}$ and $\mathcal{S'}$. Finally, from the class $\mathcal{T}$ we only have to add the graphs $T_{0,j}$ for $j \geq 0$ and $T_{1,1}$ because the $T_{i,j}$ with $i > 1$ and $j > 1$ are supergraphs of graphs in $\mathcal{B}$, the $T_{1,j}$ for $j > 0$ are supergraphs of graphs in $\mathcal{B'}$ and because for every $i,j \geq 0$, $T_{i,j} \simeq 	T_{j,i}$.
\end{proof}

Furthermore:
\begin{thm}
The graphs of \autoref{bigth} are minimal forbidden induced subgraphs for the class $\GGppm$.
\label{bigthminimal}
\end{thm}

\begin{proof}
We already proved that these graphs are forbidden, we now only need to prove that they are minimal with this respect.\\
For the graphs introduced in this section ($\mathcal{A}, \mathcal{B}, \mathcal{B'}, \mathcal{B''},
\mathcal{C}$ and $\mathcal{C'} $), the proof is rather straightforward. We only need to show that if we remove any vertex the resulting graph is no longer forbidden.\\ If we remove a \quotes{$p_i$} vertex in one path, then we disconnect the graph, and can take the symmetry of one of the two components, in terms of interval representation, so as not to have the two different types of semi-closed intervals. If we remove another vertex, then it is easy to see, through the interval representations given above, or more directly from \autoref{neighbours_of_u_lemma}, that the graph is no longer forbidden: we can shift some intervals and close one type of semi-closed intervals.\\

Now let us consider the graphs in $\mathcal{S}$, $\mathcal{S}'$, $T_{0,j}$ for $j \geq 0$, $T_{1,1}$, $R_0$, $R_1$ and $K_{2,3}^*$.
 It is immediate that $K_{2,3}^*$, $R_0$, $R_1$ and $T_{1,1}$ 
are minimal.\\ We then define $\mathcal{O} = \mathcal{S}\cup \mathcal{S'}\cup
\{T_{0,j}: j \geq 0\}$. For the graphs in $\mathcal{O}$, we know by \cite{joos} that they are minimal for the class $\GG$. But from what precedes, if we know that a graph $G$ belongs to $\GG$,  then $\mathcal{A}\cup\mathcal{B}\cup\mathcal{B'}\cup\mathcal{B''}\cup
\mathcal{C}\cup\mathcal{C'}$ is a minimal set of forbidden induced subgraph for $G$ to belong to $\GG^{\pm, +-}$ (given the fact that $G \in \GG$). So it is sufficient to show that no graph in $\mathcal{O}$ admits as an induced subgraph a graph in the previous list.\\
First, note that the graphs in the classes $\mathcal{A}, \mathcal{B}, \mathcal{B'}$ and $\mathcal{B''}$ are not induced subgraphs of graphs in $\mathcal{O}$ because the former contain two disjoint copies of $K_{1,4}^*$ and from the latter only $T_{0,j}$ also contain two such copies, but the rest of the graph does not match (in our lists of graphs, $y$ and $w$ are not adjacent). The same kind of argument applies to graphs $C_{-2}$ and $C_{-1}$. $C_1$ has two degree 5 vertices, which no graph in $\mathcal{O}$ contains. $C_0$ and $C'_{-2}$ both contain one degree five vertex which could correspond to only a few vertices in $\mathcal{S}$ and $\mathcal{S'}$, but we can see that none of this places matches the rest of our graphs. Now $C'_{-1}$, $C'_0$ and $C'_1$ have vertices of degree at least 6 which does not appear in $\mathcal{O}$. Finally $C'_2$ contains a $K_4$, which only appears in $\mathcal{S'}$ but it is easy to see that $C'_2$ is not an induced subgraph of any graph of $\mathcal{S'}$.
\end{proof}

\subsection{Algorithmical perspectives}
The proof of the previous theorem leads to some algorithms enabling us to decide if a graph is in $\UU^{\pm, +-}$ and if it is, to give a corresponding interval representation of it, as stated below. The following theorems use the standard notations $n = |V|$ and $m = |E|$.

\begin{thm}
\label{time-repr-th}
There exists an algorithm which takes a graph $G \in \mathcal{U}^{\pm, +-}$ as input, and outputs a corresponding $\mathcal{U}^{\pm, +-}$-representation of $G$ in time $O(n^2)$.
\end{thm}

\begin{proof}
We give here the algorithm, which takes a $\mathcal{U}^{\pm, +-}$-graph $G$ as input:
\begin{enumerate}[1.]
\item Prune $G$ into a twin-free graph $G'$;
\item Get a $\UU$-representation of it $\I'_0$;
\item Derive $\I'_1$ from $\I'_0$ by performing the following:
\begin{enumerate}[a.]
	\item First from right to left, try to close every open-closed interval with the transformations of the proof of \autoref{neighbours_of_u_lemma};
	\item Then from left to right, try to close similarly every closed-open interval;
\end{enumerate}
\item Let $\I'_2$ be obtained from $\I'_1$ by symmetrizing the interval representations of the connected components which contain some open-closed intervals;
\item Extend naturally the $\I'_2$ to an interval representation $\I^*$ of $G$: if $x'\in V'$ and $x \in V \setminus V'$ are twin vertices, we define $\I^*(x) = \I'_2(x)$;
\item Return $\I^*$.
\end{enumerate}
First, we claim that the algorithm is correct. Since the input graph $G$ is a $\UU^{\pm, +-}$- graph that we transform in a $\GGppm$-graph $G'$ by pruning it, we know that in each connected component  of $G'$ it is possible to close all semi-closed interval of one type. Indeed, otherwise we would have the intervals $a,b,c,d,e,u,v,w,y,z$ of \autoref{neighb-u} and \autoref{neighb-d} in a same connected component, which we proved implied that the graph is not in $\GG^{\pm, +-}$. We use the transformations of the proof of \autoref{neighbours_of_u_lemma}, which work if the semi-closed interval we try to close is not in a particular neighborhood of intervals. One important point is that, given the direction of the sweeps, if we fail to close a semi-closed interval it means that we encountered a neighborhood of the form \autoref{neighb-u} or \autoref{neighb-d} we is a certificate that the interval for $u$ or for $d$ cannot be closed by \autoref{unique_K1,4*_lemma}. Since our graph is in $\UU^{\pm, +-}$, after step 3. in each connected component there will only remain one type of semi-closed intervals. Since we want a $\UU^{\pm, +-}$ interval representation, we need step 4. Step 5. is trivial since, as we already mentioned, if $x$ and $x'$ are twin, we can give the same interval to both of them.

Concerning the time complexity, step 1 can be done in time $O(n+m)$ as in \cite{ethierry}. Step 2 can be done in $O(n^2+m)$ as shown in \cite{shuchat}. Step 3 takes time $O(n^2)$: we try to close each interval at most once, and trying to closing an interval takes $O(n)$ if we have to shift many intervals, or $O(1)$ (we check only the existence of 4 intervals at specific positions). Step 4 and 5 take time $O(n)$, hence the overall quadratic complexity, given the fact that the graphs we deal with are simple.
\end{proof}

From this algorithm we can derive another one to test if a graph is in $\UU^{\pm, +-}$, but first we state a simple lemma about the recognition of the class of mixed unit intervals. This lemma comes easily from the Algorithm 17 presented in \cite{shuchat}.

\begin{lemma}
\label{mixed unit-time}
The class of mixed unit interval graphs can be recognized in time $O(n^2)$.
\end{lemma}

\begin{proof}
The proof of this result comes from Algorithm 17 in \cite{shuchat} that we used in the previous proof. This Algorithm takes a graph $G$ supposed to be in $\UU$ and gives a $\UU$-interval representation of it. Since it works in $O(n^2)$, there exists a constant $C$ such that this algorithm performs at most $Cn^2$ operations on any graph $G \in \UU$ with $n$ vertices. Therefore one way to know whether a graph $G$ is indeed in $\UU$ is to give it to Algorithm 17, let it do at most $Cn^2$ operations and see its result: if the algorithm has terminated by that time and the interval representation it gave indeed represents $G$, we know that $G$ belongs to the $\UU$ class. In the other cases (the representation is wrong or the algorithm has not finished yet), we know that the graph $G$ is not in $\UU$. Checking that the interval representation is correct can be done in $O(n+m)$, hence our result.\\
\end{proof}

\begin{rk}
The proof of \autoref{mixed unit-time} is easy to understand because it does not require to fully dive into details of Algorithm 17 in \cite{shuchat}. However, it is possible to modify Algorithm 17 so as to recognize the class of mixed unit interval graphs in time $O(n^2)$. Indeed, it is possible to run the main loop of Algorithm 17 $n$ times. If it fails at some point because a vertex which should be present is not, or if the representation we get at the end does not correspond to the input graph, then we know that the graph is not a mixed unit interval graph. Otherwise, we know it belongs to the class. This algorithm, which does not require to find a bound $C$ mentioned in the proof, is implementable.
\end{rk}

\begin{thm}
\label{recognizability-time-th}
The class of almost-mixed unit interval graphs can be recognized in time $O(n^2)$.
\end{thm}

\begin{proof}
The proof is straightforward: we first check if the input graph $G$ is a mixed unit interval graph using \autoref{mixed unit-time}. If $G \in \UU$ then we run the algorithm of the proof of \autoref{time-repr-th}. If the interval representation we get at the end of the algorithm contains at most one type of semi-closed intervals, then this proves that our graph is an almost-mixed unit interval graph; else, the interval representation contains two types of semi-closed intervals. As our algorithm returns a correct representation when given an almost-mixed unit interval graph, then this incorrect result means that the graph is not a $\UU^{\pm, +-}$-graph. Also, by construction of the algorithm, in case the input graph is a $\UU$-graph, the output of our algorithm gives a certificate that the graph is not a $\UU^{\pm, +-}$-graph: the representation must contain in a same connected component two neighborhoods of intervals like $abcde$ and $uvxyz$ as in \autoref{neighbours_of_u_lemma}.

Testing if the output of the algorithm is of the right form can be done in time $O(n+m)$, hence the claimed complexity.
\end{proof}

\end{document}